\documentclass{ecai} 

\usepackage{latexsym}
\usepackage{amssymb}
\usepackage{amsmath}
\usepackage{amsthm}
\usepackage{booktabs}
\usepackage{enumitem}
\usepackage{graphicx}
\usepackage{color}
\usepackage[ruled, longend, noend, linesnumbered, vlined]{algorithm2e}

\usepackage{accents} 
\usepackage{array}

\newtheorem{theorem}{Theorem}

\newcommand{\BibTeX}{B\kern-.05em{\sc i\kern-.025em b}\kern-.08em\TeX}

\begin{document}


\begin{frontmatter}


\paperid{2580} 






\title{Decentralized Unlabeled Multi-agent Pathfinding Via Target And Priority Swapping (With Supplementary)}


\author[A, B]{\fnms{Stepan}~\snm{Dergachev}\orcid{0000-0001-8858-2831}\thanks{Corresponding Author. Email: dergachev@isa.ru\\ This is a pre-print of the paper accepted to ECAI 2024. Its main body is similar the camera-ready version of the conference paper. In addition this pre-print contains Supplementary Material incorporating extended empirical results and analysis.}}
\author[A, B, C]{\fnms{Konstantin}~\snm{Yakovlev}\orcid{0000-0002-4377-321X}}

\address[A]{Federal Research Center ``Computer Science and Control'' of the Russian Academy of Sciences}
\address[B]{HSE University}
\address[C]{AIRI}


\begin{abstract}
In this paper we study a challenging variant of the multi-agent pathfinding problem (MAPF), when a set of agents must reach a set of goal locations, but it does not matter which agent reaches a specific goal -- Anonymous MAPF (AMAPF). Current optimal and suboptimal AMAPF solvers rely on the existence of a centralized controller which is in charge of both target assignment and pathfinding. We extend the state of the art and present the first AMAPF solver capable of solving the problem at hand in a fully decentralized fashion, when each agent makes decisions individually and relies only on the local communication with the others. The core of our method is a priority and target swapping procedure tailored to produce consistent goal assignments (i.e. making sure that no two agents are heading towards the same goal). Coupled with an established rule-based path planning, we end up with a TP-SWAP, an efficient and flexible approach to solve decentralized AMAPF. On the theoretical side, we prove that TP-SWAP is complete (i.e. TP-SWAP guarantees that each target will be reached by some agent). Empirically, we evaluate TP-SWAP across a wide range of setups and compare it to both centralized and decentralized baselines. Indeed, TP-SWAP outperforms the fully-decentralized competitor and can even outperform the semi-decentralized one (i.e. the one relying on the initial consistent goal assignment) in terms of flowtime (a widespread cost objective in MAPF).

\end{abstract}

\end{frontmatter}


\section{Introduction}
\label{sec:intro}

Multi-agent navigation is a vital and non-trivial problem which arises in various practical applications such as mobile robotics, transportation systems, video-games etc. Generally, the problem asks to find a set of non-colliding trajectories (paths) for a group of agents operating in a shared workspace. Numerous modifications, setups and approaches for this problem exist. One of the most well-studied setups is when each agent is asked to reach a specific goal location, i.e., the assignment of goals to agents is given as the problem input~\citep{stern2019multi}. Another variant is the so-called, unlabeled or anonymous multi-agent pathfinding (AMAPF)~\citep{yu2013multi}. In this setting it is assumed that the agents are interchangeable in a sense that for a single agent there is no strict requirement to achieve a particular goal. It is this problem that we focus on in this work.


\begin{figure}[t]
    \centerline{\includegraphics[width=0.95\columnwidth]{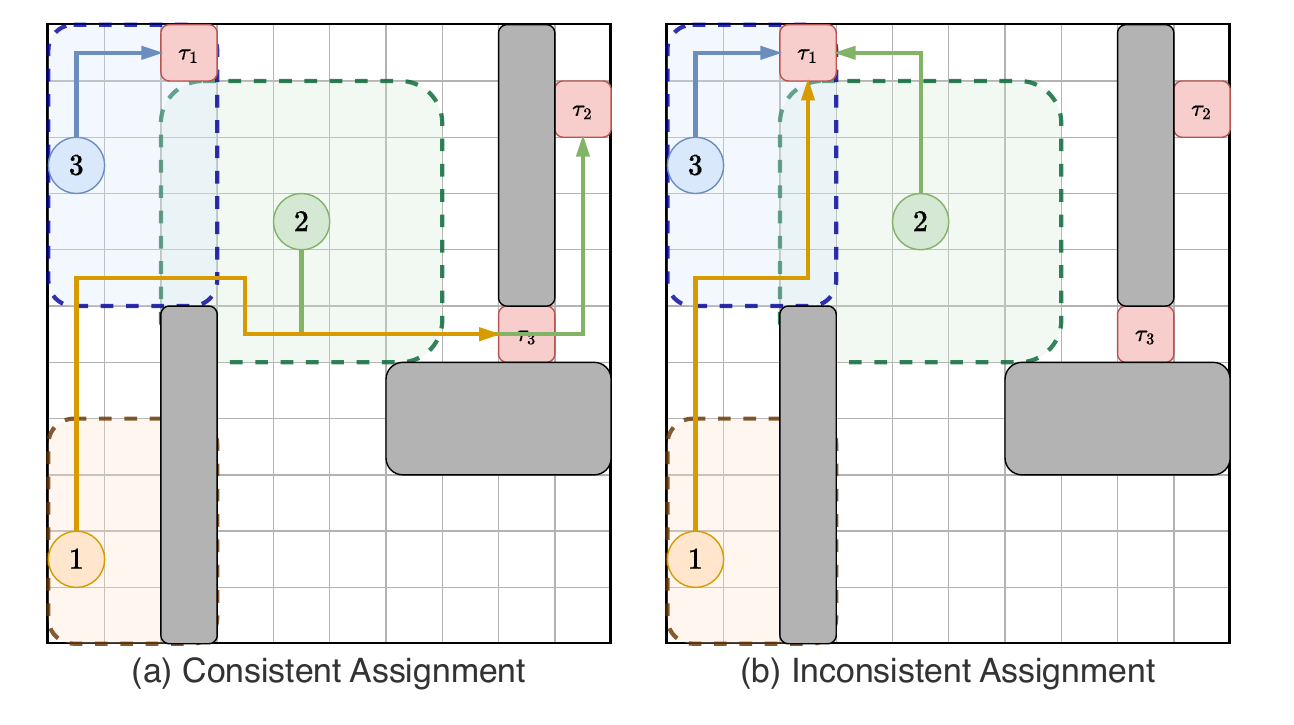}}
    \caption{An example of a decentralized AMAPF instance with a consistent (left side) and an inconsistent (right side) initial assignments. Solid circles depict agents. Red squares are the goals. The colored area around each agent is its communication zone (it is able to communicate with the others only if they reside inside this zone).\vspace*{20 pt}} 
    \label{fig:vis_abs}
\end{figure}

Numerous methods have been recently proposed to solve MAPF. Some of them are intended to find optimal solutions w.r.t. space-time discretization (CBS~\citep{sharon2015conflict}, M*~\citep{wagner2011m}, ICTS~\citep{sharon2013increasing} to name a few), while the others trade-off optimality for lower runtimes in a controlled fashion, like ECBS~\citep{barer2014suboptimal}, EECBS~\citep{li2021eecbs}, ODrM*~\citep{ferner2013odrm}, or completely ignore the cost objective in favor of smaller runtime and scalability like Push and Rotate~\citep{de2013push}, PIBT~\citep{okumura2022priority} and others. The same applies to AMAPF, where both optimal solvers~\citep{yu2013multi} and more scalable suboptimal solvers~\citep{okumura2023solving} do exist.

Still, most of the state-of-the-art (A)MAPF solvers intrinsically assume that there exist a centralized controller that fully observes the environment and is in charge of constructing plans that the agents need to execute. In practice, however, deploying such centralized systems may be costly and decentralized methods to tackle (A)MAPF are desirable. In this case each agent has to decide on its own, based on the limited observation/communication, how to choose the target and move towards it while avoiding the collisions.

Most of the recent decentralized MAPF solvers are learning-based: PRIMAL~\citep{sartoretti2019primal}, G2RL~\citep{wang2020mobile}, SCRIMP~\citep{wang2023scrimp}, FOLLOWER~\citep{skrynnik2024learn} to name a few. Surprisingly, the decentralized AMAPF is much less studied and to the best of our knowledge no established decentralized AMAPF solver exists. Our work aims to fill this gap.

We start with a prominent (suboptimal) centralized AMAPF solver, i.e. TSWAP~\citep{okumura2023solving}, and show how can one build a semi-decentralized and then fully decentralized AMAPF algorithm on top of it. The main bottleneck in doing so is resolving inconsistencies in individual goal assignments, i.e. dealing with the problem when several agents are heading towards the same goal -- see Fig.~\ref{fig:vis_abs}. We introduce both a naive way to cope with it and a more involved one, that relies only on the local communication and is based on a specific target and priority swapping procedure. We prove that the most advanced version of our solver, which we call TP-SWAP, is complete, i.e. it guarantees that each agent will arrive to a unique goal, thus all goals will be reached.

We conduct a thorough empirical evaluation of the proposed decentralized  AMAPF solvers and compare them with the centralized baseline, i.e. TSWAP. We show that our most enhanced algorithm, i.e. TP-SWAP, indeed, outperforms the other suggested decentralized variants. Moreover, its flowtime (one of the widely used measures of the solution cost) is consistently lower (better) compared to the centralized TSWAP with the random goal assignment, and its makespan (another widely used measure of the solution cost) approaches the one of the latter when the number of agents increases. Overall, our findings pave the way to creating efficient fully decentralized AMAPF solvers that rely only on the local communication/observations, while still guarantying completeness.

\section{Related Works}
\label{sec:related_works}

\textit{MAPF} is a well-studied problem with a large number of different formulations~\citep{stern2019multi}. The dominant number of papers devoted to \textit{MAPF} assume the existence of a centralized controller and rely on a discretized representation of the workspace (e.g. a grid).

The classical \textit{MAPF} problem formulation assumes that each agent is assigned a specific target to reach. It is known that finding an optimal solution, whether in terms of makespan or flowtime, is an NP-hard problem~\citep{yu2013multi}. The techniques to solve \textit{MAPF} are plentiful. Some of them are aimed at obtaining provably optimal~\citep{sharon2015conflict} or bounded-suboptimal~\citep{barer2014suboptimal} solutions. These methods typically do not scale well to a large number of agents. On the other hand, if it is necessary to quickly find a solution and the cost is not of utmost importance, then the rule-based solvers can be applied~\citep{de2013push}. A possible compromise between the solution cost and the performance may be provided by the prioritized planning~\citep{vcap2015prioritized}, which often finds close-to-optimal solutions and is also fast and scalable. However, prioritized planning is incomplete in general.

Another variant of \textit{MAPF} is the \textit{anonymous/unlabeled MAPF (AMAPF)}~\citep{stern2019multi,okumura2023solving,yu2013multi}, when the goals are not assigned to the agents initially. Unlike classical MAPF, AMAPF is always solvable~\citep{okumura2023solving}. Similarly to MAPF, most of the AMAPF solver are centralized. One of the key approaches is method~\citep{yu2013multi}, that allows finding makespan-optimal solutions in polynomial time (in contrast to the classical MAPF), although making practical problems on large graphs is inefficient. The paper~\citep{ali2024improved} addresses this limitation and proposes improvements to find solutions more efficiently. Alternatively, there is the fast TSWAP solver~\citep{okumura2023solving}, which allows obtaining suboptimal solutions.



The number of methods that consider both decentralized scenarios and unlabeled case is very limited. An adaptation of the centralized method of~\citep{turpin2014capt} to the decentralized setting was presented in the same paper. However, as was shown in~\citep{panagou2019decentralized}, it does not guarantee the absence of collisions between the agents. The latter paper presents another decentralized method that is based on consideration of different number of potential goal assignments. This number can be prohibitively large. Moreover, the algorithm itself is not suited to operate in the non-empty environments.


Finally, a rapidly evolving research line is the one that suggests utilization of deep learning and multi-agent reinforcement learning for both centralized~\citep{ji2021decentralized} and decentralized navigation and goal assignment, see~\citep{lowe2017multi,khan2019learning, khan2021large}. However, the learnable methods are not typically able to provide any sorts of guarantees, require extensive training and often perform poorly on the problem instances that are not alike the ones used for training.


\section{Problem Statement}
\label{sec:problem_statement}



We first present the centralized variant of the problem and then switch to the decentralized one. 

\paragraph{Centralized AMAPF} Consider a set $\mathcal{N}$ of $n$ agents, each confined to a connected, undirected graph $\mathcal{G} = (\mathcal{V}, \mathcal{E})$. There is a mapping $s : \mathcal{N} \rightarrow \mathcal{V}$ that assigns each agent to a specific start vertex, and a set $\mathcal{T} \subset \mathcal{V}$ of $n$ target/goal vertices.


Time is discretized into timesteps. At each timestep, an agent can choose either to move to an adjacent vertex (a \emph{move} action) or to remain at its current vertex (a \emph{wait} action). A path for an agent $i$ from vertex $v \in \mathcal{V}$ to vertex $v' \in \mathcal{V}$, denoted by $\pi^i(v, v')$, is defined as a sequence of actions that takes the agent from $v$ to $v'$. The cost of the path is determined by the timestep at which the agent reaches its final destination. Additionally, we assume that once an agent reaches its target, it remains there and waits.

Paths should not include two types of conflicts:
\begin{itemize}
    \item \textit{Vertex conflict}: occurs between the agents $i, j \in \mathcal{N}$ \textit{iff} they stay at the same vertex at the same timestep.
    \item \textit{Swapping conflict}: occurs between the agents $i, j \in \mathcal{N}$ \textit{iff} they traverse the same edge at the same timestep.
\end{itemize}



\textit{The problem} is to find a sequence of actions (a path) for each agent such that \emph{(i)} each individual path for agent $i$ starts at the predefined start location $s(i)$ and ends at one of the predefined goals $\tau \in \mathcal{T}$; \emph{(ii)} all goal locations are reached; and \emph{(iii)} all pairs of paths are conflict-free.

The quality of an AMAPF solution is typically evaluated using either \emph{flowtime} or \emph{makespan}, with lower values indicating better solutions. Flowtime is the sum of the costs of all paths in the solution, while makespan is the maximum cost among these paths. In this work, we do not impose a strict requirement to optimize the cost of the solution, but naturally, solutions with lower costs are preferable.

\paragraph{Decentralized AMAPF} In a decentralized setup, each agent independently decides on its actions at each timestep, based on the limited information it obtains through local observation and communication. We assume that each agent has knowledge of the entire graph and can exchange information with other agents located within a distance of $k$ edges from its current vertex. In our experiments, we use grid environments where the communication range is defined by a $(2k+1) \times (2k+1)$ cell area with the agent positioned at the center.

Moreover, we allow for chain communication between agents. This means that if agent $i$ is within the communication range of agent $j$, and agent $j$ is within the communication range of agent $k$, then agent $i$ can exchange information with agent $k$ through agent $j$, and vice versa.

For the purposes of this study, we abstract away from communication issues and assume instantaneous, error-free information exchange.




\section{Methods}
\label{sec:suggested_method}


Our decentralized solver, TP-SWAP, is developed on the basis of the rule-based centralized method, TSWAP~\cite{okumura2023solving}. Therefore, we begin by explaining TSWAP and then gradually explore how it can be adapted for the decentralized seeting.

\subsection{TSWAP}
\label{sec:tswap}



TSWAP~\citep{okumura2023solving} solves the AMAPF problem in two stages. In the first stage, it creates an initial \textit{consistent goal assignment}, i.e. the one where each goal is uniquely assigned to a single agent, ensuring that no two agents share the same goal. In the second stage, the algorithm iteratively moves the agents toward their assigned goals and, if necessary, reassigns goals between them while always maintaining the consistency of the goal assignment.

\paragraph{Initial Goal Assignment} In general, TSWAP can handle any consistent goal assignment. In their work, the authors of TSWAP explored several methods for initial goal assignment and evaluated them empirically. For our experiments, when using TSWAP as a centralized baseline, we adopted the assignment method that demonstrated the most promising results in the original paper.

\paragraph{Moving towards the goals with target swapping} 


At each planning iteration, TSWAP sequentially examines all agents. For each agent $i$, it identifies the current vertex $v$ and deterministically selects the next vertex $v'$ based on the shortest path to its goal. If vertex $v'$ is free, then it is marked as occupied by agent $i$.

If vertex $v'$ is already occupied by another agent $j$, the agent $i$ picks a wait action and checks two possible cases. First, if vertex $v'$ is the target of agent $j$, agents $i$ and $j$ swap their targets, as illustrated in Fig.~\ref{fig:tswap}-a. Otherwise, TSWAP checks whether agent $i$ is involved in a deadlock. A deadlock occurs when a loop sequence of agents (including $i$) is formed, such that each agent's next vertex in their shortest path is currently occupied by the next agent in the sequence. If a deadlock is detected, the targets of the agents within the sequence are rotated, meaning each agent is reassigned the target of the next agent in the loop. This scenario is depicted in Fig.~\ref{fig:tswap}-b.

After the planning iteration is complete, the algorithm moves each agent to its designated vertex, if necessary.

\begin{figure}[t]
    \centerline{\includegraphics[width=0.85\columnwidth]{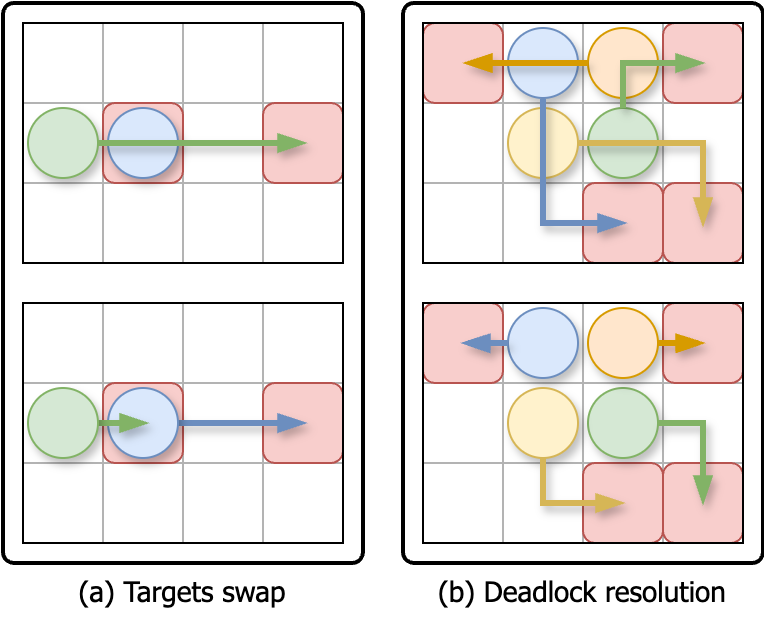}}
    \caption{Examples of conflict and deadlock resolutions in TSWAP: (a) Illustration of the target-swapping mechanism when an agent occupying its target location blocks the path for another agent. (b) Illustration of the deadlock resolution mechanism, where a sequence of agents forms a loop, causing them to block each other's paths. \vspace*{20 pt}} 
    \label{fig:tswap}
\end{figure}


\subsection{Decentralized TSWAP with Consistent Target Assignment}
\label{subsec:dec_tswap_cons}





TSWAP can, in principle, be adapted to a decentralized setup where a centralized controller is not present, and agents have limited communication capabilities within a certain range. However, an initial consistent target assignment is still mandatory.

At each step, each agent can execute an iteration of the TSWAP algorithm using only the information available within its communication range. It is important to note that all agents within the same subgroup have access to the same information, ensuring that the outcome of the algorithm's execution will be consistent across all members of the group.

To correctly perform the TSWAP iteration, an agent needs specific information. First, all members of the subgroup must examine the agents in the same order. Therefore, each agent needs to know the priority $pr$ of the other members. We assume that each agent $i$ is identified by a unique number (e.g., the serial number of the robot). For simplicity, we will assume that the agent, its identifier, and its priority are identical, i.e., $pr = i$.

Second, each agent in the subgroup must be aware of the positions of the other subgroup members and their assigned targets. This information is essential for avoiding vertex conflicts and for performing target-swapping or deadlock resolution.

Finally, we must establish the communication conditions that are sufficient to safely and correctly execute the TSWAP iteration. Specifically, we need to ensure that for each agent $i$ and its next vertex $v'$, no other agent $j$ is located at $v'$ or plans to move there in the next step. To meet this requirement, each agent must have information about other agents within a range of at least two vertices from its own location. Therefore, the minimum necessary communication range must cover this distance.

Additionally, to detect and resolve deadlocks, an agent must have information about any chain of adjacent agents to which it belongs. This condition is satisfied if chain communication is allowed, as discussed in the section on the decentralized scenario in Section~\ref{sec:problem_statement}.

It is important to note that if the initial assignment is inconsistent, the described variant of TSWAP may fail to solve the instance. Thus, it can arguably be considered \emph{semi}-decentralized, as achieving a consistent initial goal assignment requires some form of centralization, such as global information sharing. 

In the next section, we will focus on fully decentralized AMAPF solvers (based on TSWAP) that do not require the initial assignment to be consistent.

\subsection{Naive Fully-decentralized TSWAP}
\label{subsec:naive}




In a fully decentralized setting, each agent must independently choose its goal. This might be done randomly, or each agent might select the closest goal to its starting location. However, in such cases, multiple agents may head toward the same goal, while some goals might remain unassigned. A naive approach to resolve these inconsistencies involves memorizing which goals have already been achieved by other agents and selecting a different goal from those not on this list if necessary.

Specifically, at each step, every agent $i$ must verify through local communication that its current goal (the one it is heading toward) is not already occupied by another agent. If it is, the agent should add this goal to a dedicated list of occupied goals (initialized as $\emptyset$) and then select another goal (randomly or according to a specific rule) from those not included in this list. Additionally, agents can exchange their lists of occupied goals to increase the likelihood that each agent will choose an unoccupied goal.

While this approach can indeed restore consistency in goal assignments and resolve conflicts where multiple agents are heading toward the same goal, it does not fully utilize the potential of local information exchange. As a result, agents may still need to explore all the targets individually until finding an unoccupied one. To address this limitation, we propose TP-SWAP -- an improved, fully decentralized AMAPF solver that allows agents to select goals in a more informed manner.

\begin{algorithm}[!t]
    \caption{TP-SWAP Algorithm}
    \footnotesize
    \label{alg:main_tpswap}
    \SetKwInput{Input}{Input}
    \SetAlgoNlRelativeSize{-1}
    \newcommand{\mycapfn}[1]{\textsc{\scriptsize #1}}\SetFuncSty{mycapfn}
    \SetKwFunction{selectTarget}{selectTarget}
    \SetKwFunction{generateUniquePriority}{generateUniquePriority}
    \SetKwFunction{identifyLocalAvailableAgents}{identifyLocalAvailableAgents}
    \SetKwFunction{recievePosition}{recievePosition}
    \SetKwFunction{getTarget}{getTarget}
    \SetKwFunction{getPriority}{getPriority}
    \SetKwFunction{TPUPDATE}{TP-UPDATE}
    \SetKwFunction{moveTo}{moveTo}

    \KwIn{$i$ -- agent unique id; $\mathcal{G}$ -- graph; $\mathcal{T}$ -- set of all goals.}

    \BlankLine
    $\tau \gets$ \selectTarget{$\,\mathcal{G}$, $\mathcal{T}$}\; \label{line:select_goal}
    $pr \gets$ \generateUniquePriority{$\,i$}\; \label{line:chose_prior}
    $TP[\tau'] \gets -\infty \; \forall \tau' \in \mathcal{T}$; $TP[\tau] \gets pr$\;\label{line:create_tp}

    \While{\normalfont{{\textsc{true}}}}{ \label{line:move_begin}
        $NH \gets$ \identifyLocalAvailableAgents{}\; \label{line:communication_begin}
        $TA \gets \{\}; PR \gets \{\}; V \gets \{\}$\; \label{line:create_tables}
        $V[j] \gets$ \recievePosition{$\,j$} $\forall j \in NH$\; \label{line:get_st}
        $TA[j] \gets$ \getTarget{$\,j$} $\forall j \in NH$\; \label{line:get_target}
        $PR[j] \gets$ \getPriority{$\,j$} $\forall j \in NH$\; \label{line:get_pr}
        $TP[\tau] \gets \max\limits_{j \in NH}\,{TP^j[\tau]}\,\forall j \in NH$\; \label{line:get_tp}
        $\tau', pr', v', TP' \gets$\TPUPDATE{$i, NH, V, TA, PR, TP, \mathcal{T}$} \; \label{line:tp_update}
        $\tau \gets \tau'$; $pr \gets pr'$; $TP \gets TP'$\; \label{line:update_priority}\label{line:update_goal}
        \moveTo($v'$)\; \label{line:move}
    } \label{line:move_end}
\end{algorithm}

\subsection{TP-SWAP: Target-Priority Swapping For Decentralized AMAPF}
\label{subsec:nav_pipe}

The approach used to enhance the previously described fully decentralized AMAPF solver is based on two key ideas. First, it is advantageous not only to identify and memorize already occupied goals (and possibly share this information) but also to track the desired goals of other agents. Second, agent priorities can be utilized to restore consistency in goal assignments, with the possibility of exchanging these priorities, as will be explained later.

To achieve this, each agent individually maintains a \textit{target-priority} assignment table, $TP$, which is a mapping from agent priorities (not identifiers!) to targets. The $TP$ table is maintained throughout the entire process and is used by agents to resolve conflicts in their current goal assignments. 

In addition to the target-priority mapping, agents also temporarily create and share information at each time step regarding their current locations (graph vertices), targets, and priorities. To facilitate this, temporary tables/dictionaries $V$, $TA$, and $PR$ are introduced, which are generated from scratch at each iteration of the algorithm.

The general outline of TP-SWAP is presented in Algorithm~\ref{alg:main_tpswap}. Similar to the naive approach, each agent individually selects its target (line~\ref{line:select_goal}) and determines its priority $pr$ (line~\ref{line:chose_prior}) before starting the movement.

Initially, the value for each target $\tau \in \mathcal\{T\}$ in the table is set to $TP[\tau] = -\infty$. When an agent selects a goal $\tau \in \mathcal{T}$, it updates the corresponding entry in the table to $TP[\tau] = pr$ (line~\ref{line:create_tp}). This indicates that, at the start, the agent only has information about its own goal.

An iterative process then begins, guiding the agent toward its goal. Each iteration starts by identifying the agents available for communication (line~\ref{line:communication_begin}) and gathering information about them. The agent collects information on the locations ($V$), current targets ($TA$), and current priorities ($PR$) of the members in the connected subgroup (lines~\ref{line:create_tables}-\ref{line:get_pr}).

These tables are populated with up-to-date data relevant to the current group to which the agent belongs, and they are updated at each step. The position of a neighboring agent $j$ at the current timestep is denoted as $V[j]$, its target as $TA[j]$, and its priority as $PR[j]$.

Next, the agent updates its $TP$ table using the collective knowledge of all subgroup members. If the agent receives information that a target $\tau$ was selected by another agent (who may not be part of the current subgroup) with a priority $pr' \in \mathbb{N}$ higher than what is currently recorded in the table ($pr' > TP[\tau]$), the agent updates the table to $TP[\tau] = pr'$ (line~\ref{line:get_tp}).

Once all necessary information has been gathered, the agent initiates the procedure for resolving assignment conflicts, updating targets and priorities, and selecting the next vertex. This is done by a (core) routine named \textsc{TP-Update}(line~\ref{line:tp_update}). It eliminates inconsistencies in goal assignment within the current subgroup and prevents collisions, similarly to the TSWAP algorithm. A more detailed description of the target-priority update procedure is provided in the next section.

Finally, the algorithm updates the agent's current state and moves the agent to its next determined location (lines~\ref{line:update_goal}-\ref{line:move}).

\subsubsection{TP-UPDATE: Procedure to resolve conflicts in TP-table}
\label{subsec:goal_exch}

\begin{algorithm}[!t]
\caption{TP-UPDATE Procedure}
\footnotesize
\label{alg:tp_update}
\SetKwInput{Input}{Input}
\SetAlgoNlRelativeSize{-1}
\newcommand{\mycapfn}[1]{\textsc{\scriptsize #1}}\SetFuncSty{mycapfn}
\SetKwFunction{sortDecreasingPriorities}{sortDecreasingPriorities}
\SetKwFunction{nextVertex}{nextVertex}
\SetKwFunction{swap}{swap}
\SetKwFunction{inDeadlock}{inDeadlock}
\SetKwFunction{getDeadlockSequence}{getDeadlockSequence}
\SetKwFunction{rotateTargetsPriorities}{rotateTargetsPriorities}

\KwIn{$i$ -- agent unique id; $NH$ -- current subgroup; $V$ -- positions of agents in the current group; $TA$ -- target assignment of the current group; $PR$ -- table of priorities for agents in the current group; $TP$ -- table of targets and corresponding agents' priorities; $\mathcal{T}$ -- set of all targets; $\mathcal{G}$ -- graph.}
\BlankLine
$NH \gets$ \sortDecreasingPriorities{$NH$, $PR$}\; \label{line:sort_agents}

\For{$j \in NH$}{ \label{line:incons_begin}
    \If{$TP[TA[j]] > PR[j]$}{ \label{line:incons_cond}
        $TA[j] \gets \tau' \in \mathcal{T}$ s.t. closest to $V[j]$, $TP[\tau'] \leq PR[j]$\; \label{line:incons_new_goal}
        $TP[TA[j]] \gets PR[j]$\; \label{line:new_goal_prior}
    }
} \label{line:incons_end}

\For{$j \in NH$}{ \label{line:tswap_begin}
    \If{$S[j] = TA[j]$}{
        \textbf{continue}\;
    }
    $v \gets$ \nextVertex{$V[j]$, $TA[j]$, $\mathcal{G}$}\;
    \If{$\exists k \in NH$ s.t. $V[k] = v$}{
        \If{$v = TA[k]$}{
            \swap{$TA[j]$, $TA[k]$}\;
            \swap{$PR[j]$, $PR[k]$}\; \label{line:swap_pr}
        }
        \ElseIf{\inDeadlock{$j$, $V$, $TA$}}{
            $D \gets$ \getDeadlockSequence{$\,j$, $V$, $TA$}\;  
            \rotateTargetsPriorities{$D$, $TA$, $PR$}\; \label{line:roatate_pr}
        }
    }
    \Else{
        $V[j] \gets v$\;
    }
} \label{line:tswap_end}

\KwRet $TA[i], PR[i], V[i], TP$\; \label{line:goals_update_return}

\end{algorithm}

The most critical component enabling the correct execution of the suggest algorithm is the target-priority update procedure, which is described below\footnote{It is important to note that the operations described below are performed solely based on information available to a specific group of agents that are able to communicate and share information with each other.}.

The \textsc{TP-UPDATE} procedure pursues two goals. First, it resolves assignment conflicts and updates the target-priority table, ensuring that the target assignment within the current connected subgroup is consistent. Second, it aligns the target-priority assignment of the current group with the available information. This alignment prevents any agent in the group from being assigned a target that has already been selected by another agent with a higher priority. This rule is enforced even if the higher-priority agent is not part of the current group, provided that its information is present in the $TP$ table.

This approach not only ensures the correct execution of the TSWAP algorithm but also prevents agents from selecting targets already claimed by higher-priority agents, thereby maintaining global consistency.

Let's examine Alg.~\ref{alg:tp_update} in more detail. To implement the key ideas, all agents within the current group are considered and sorted according to their priority (lines~\ref{line:sort_agents}-\ref{line:incons_end}). If the $TP$ table indicates that the target of a particular agent $i$ has been selected by another agent with a higher priority, then the agent $i$ will select a new target (line~\ref{line:incons_new_goal}).

The new target is selected based on the following rules: \textit{(i)} it must be the closest available target to the agent, and \textit{(ii)} the priority recorded for it in the $TP$ table must be lower than the current priority of the agent.  In Section~\ref{sec:theoretical}, we will demonstrate that there always exists a target satisfying these conditions.

After the $TP$ table is updated (along with the auxiliary $TA$ and $PR$ tables), these updated tables are used as input for the iteration of the TSWAP algorithm (lines~\ref{line:tswap_begin}-\ref{line:tswap_end}). Additionally, if two agents exchange targets during this TSWAP step, they must also exchange their priorities (lines~\ref{line:swap_pr},~\ref{line:roatate_pr}).

As a result, an agent gets an updated table $TP$ along with the updated priority, target and a vertex to move to (line~\ref{line:goals_update_return}).

\subsubsection{Running Example}

\begin{figure}[t]
    \centerline{\includegraphics[width=0.8\columnwidth]{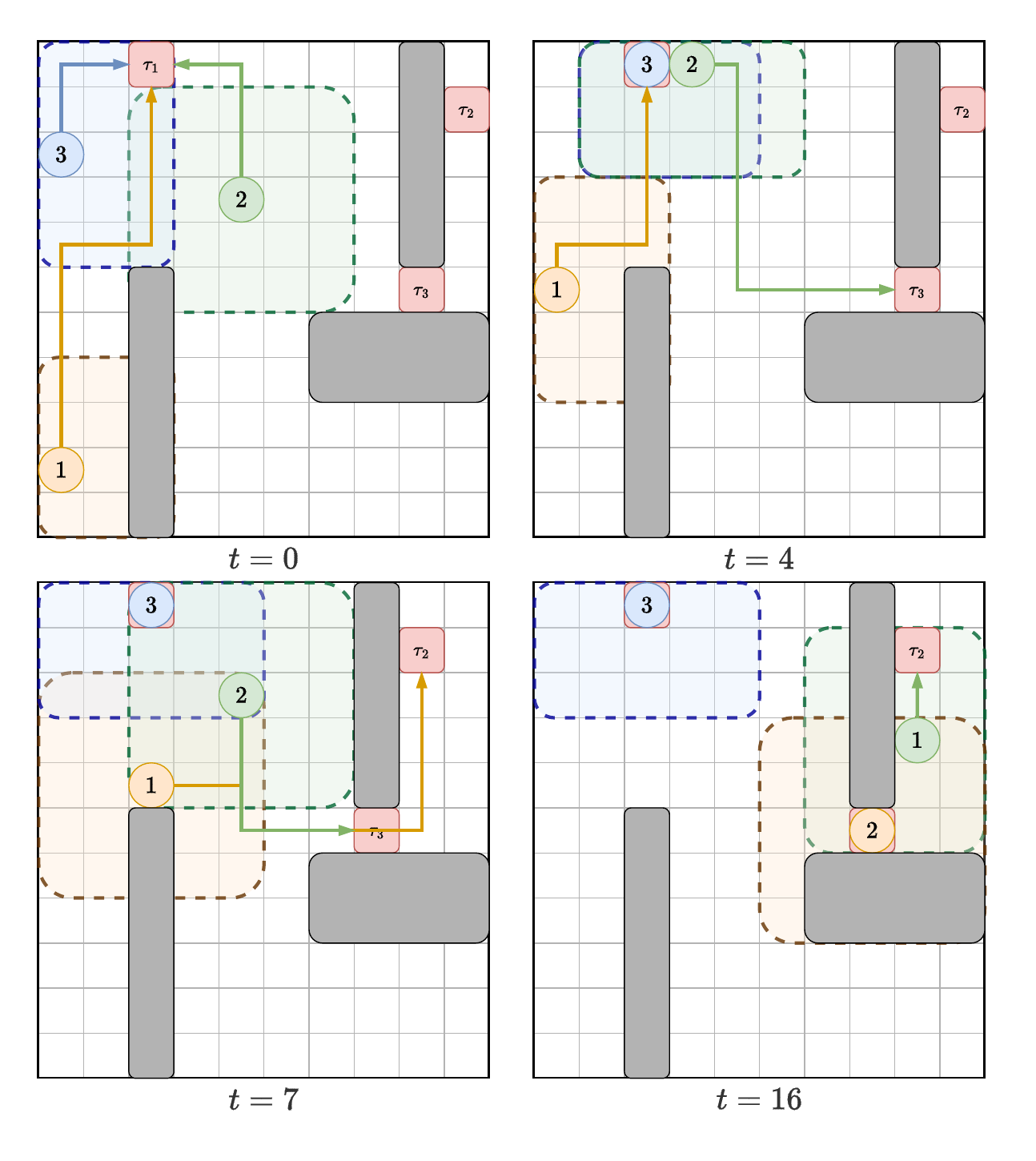}}
    \caption{An example of solving a decentralized AMAPF instance. Agents are depicted as disks (with the number showing their current priority). The dashed lines illustrate the agents' communication ranges. The red cells represent the goals that the agents need to reach. \vspace*{20 pt}} 
    \label{fig:dec_tswap_demo}
\end{figure}

Let's examine an example of solving an AMAPF problem in a decentralized fashion, as illustrated in Fig.~\ref{fig:dec_tswap_demo}.

Here, three agents (shown as blue, green, and orange disks) are confined to a grid. The goals (known to each of the agents) are depicted as the red squares labeled as $\tau_1$, $\tau_2$, $\tau_3$. The communication range of each agent is marked by the dashed lines matching the agents' colors.

At time $t = 0$, each agent selects the closest target and plans a path towards it, shown by the colored arrow. Initially, all agents choose the same target, $\tau_1$, and assign themselves priorities, displayed inside the circles. The agents then begin moving toward their selected targets.

By time $t = 4$, the blue agent has reached $\tau_1$ and is within communication range of the green agent. The green agent, having a lower priority, switches to a new target, $\tau_3$, and recalculates its path. Meanwhile, the orange agent continues moving toward $\tau_1$, unaware that a higher-priority agent has already reached it.

At time $t = 7$, the orange and green agents meet, and the orange agent receives information that both $\tau_1$ and $\tau_3$ have been claimed by higher-priority agents. The orange agent then chooses the remaining target, $\tau_2$.

Between $t = 7$ and $t = 16$, the green agent reaches $\tau_3$ before the orange agent can reach $\tau_2$, blocking the orange agent's path. To resolve the conflict, the orange and green agents swap their goals and priorities. By time $t = 16$, the orange agent has successfully reached $\tau_3$, while the green agent is now heading toward $\tau_2$.

\subsubsection{Theoretical Analysis}
\label{sec:theoretical}

\begin{theorem}
\label{theorem:complete}
There exists a finite time $t$ at which all goals will be achieved by the agents utilizing Alg.~\ref{alg:main_tpswap}.
\end{theorem}
\begin{proof}


Consider the following function $\phi$:


\begin{equation}
\begin{split}
\phi(t) = &\phi_1(t) + \phi_2(t) + C \cdot \phi_3(t) \\
\phi_1(t) = &\sum_{i \in \mathcal{N}} \; dist(\,\overline{V_t}[\,i\,], \, \overline{TA_t}[\,i\,]\,) \\
\phi_2(t) = &\sum_{i \in \mathcal{N}} \big| \{ \, j \, : \, j \in \mathcal{N}, \; \overline{TA_t}[\,j\,] \in \Pi(\overline{V_t}[\,i\,], \, \overline{TA_t}[\,i\,])\,\} \big| \\
\phi_3(t) = &\sum_{ pr \in \overline{PR_t}} \big| \{ \, \tau \, : \, \tau \in \mathcal{T}, \; TP^{pr}_t [\,\tau\,] \leq pr \, \} \big|
\end{split}
\end{equation}
\noindent where 
\begin{itemize}
    \item $dist(v, v')$ -- the shortest path length between vertices $v, v' \in \mathcal{V}$,
    \item $\Pi(v, v')$ -- a set of vertices in the shortest path between $v, v' \in \mathcal{V}$, 
    \item $\overline{V_t}$ -- the table of agents' positions on the graph at time $t$
    \item $\overline{TA_t}$ -- the global goal assignment at time $t$
    \item $\overline{PR_t}$ -- the global priority assignment at time $t$
    \item $TP^{pr}_t$ -- the target-priority table at time $t$ of agent with priority $pr$
    \item $C$ -- finite scalar value, at least equal $\bigl( 2\max\limits_{v, v' \in \mathcal{V}} dist(v, v') + 1 \bigr)$
\end{itemize}

The function consists of the following components:
\begin{itemize}
    \item $\phi_1(t)$ represents the total distance from each agent's current position to its assigned target.
    \item $\phi_2(t)$ counts for each agent the number of other agents that lie in its path.
    \item $\phi_3(t)$ counts for each agent the number of targets that agent either knows nothing about or knows are chosen by lower-priority agents.
\end{itemize}

Let's demonstrate that there exists a specific moment in time when all agents will have successfully achieved their goals. To establish this, it suffices to show that \textit{(i)} the TP-UPDATE procedure is correct, i.e. it will always find consistent target assignment for a connected subgroup of agents \textit{(ii)} the function $\phi$ is bounded from below, \textit{(iii)} $\phi$ is decreasing, and \textit{(iv)} $\phi$ reaching the lower bound can occur \emph{iff} all agents have achieved their targets.

\paragraph{TP-UPDATE Correctness} It can be observed that Alg.~\ref{alg:tp_update} (lines~\ref{line:incons_begin}-\ref{line:incons_end}) creates and maintains a consistent goal assignment $TA$ within a subgroup of agents. Crucially, each agent is guaranteed to find a new target if it must abandon its current one (line~\ref{line:incons_new_goal} of Alg.~\ref{alg:tp_update}).

To prove this, let's assign a new numbering to the agents, reversed according to their priorities (i.e., the agent with the highest priority is numbered 1, and the last one is $n$). We will prove by induction that an agent $k$ can reject no more than $k-1$ targets.



\begin{itemize}
    \item \textbf{Base case:} the first agent can reject no targets (see lines~\ref{line:incons_begin}-\ref{line:incons_end} of Alg.~\ref{alg:tp_update}).
    \item \textbf{Induction step:} assume that an agent $k$ can reject no more than $k-1$ targets. Now, consider an agent $k+1$. Suppose this agent can reject more than $k$ targets. Among these rejected targets, at least two must have been targets that were not rejected by agent $k$. If agent $k+1$ rejected these goals, it means they were at some point chosen by agents with priorities higher than $k+1$. Since the priority associated with each goal does not decrease, there are two possibilities: either both of these goals were achieved by agent $k$ (which is impossible), or one of the targets must have been achieved by an agent with a higher priority than $k$. This would mean agent $k$ would have also had to reject that target. Hence, this is a contradiction.

\end{itemize}

\paragraph{Boundedness} Functions $\phi_1$ and $\phi_2$ are each bounded from below by 0 for any assignment, while $\phi_3$ is the sum of a finite set of non-negative integers. This implies that $\phi_3$ is also bounded from below. Consequently, the overall function $\phi$ is bounded below.

\paragraph{Monotonicity} 

If a consistent target assignment is provided as input to Alg.~\ref{alg:tp_update} and remains unchanged between lines~\ref{line:incons_begin} and~\ref{line:incons_end}, the function $\phi_1(t) + \phi_2(t)$ strictly decreases, as established in the analysis of TSWAP in \citep{okumura2023solving}.

The function $\phi_3(t)$ is non-increasing because, at each timestep, the algorithm updates $TP$ tables, ensuring the known number of targets selected by lower-priority agents, or left unselected, either stays the same or decreases (line~\ref{line:get_tp} of Alg~\ref{alg:main_tpswap}).

If $\phi_3$ remains unchanged at a timestep, so $\phi_1 + \phi_2$ decreases unless all targets are achieved. Otherwise, if the subgroup configuration changes, $\phi_3$ decreases. Although $\phi_1 + \phi_2$ may increase, the magnitude of change in $(C \cdot \phi_3)$ is always greater, ensuring $\phi$ decreases overall.

Since agent subgroups do not intersect, $\phi$ can be expressed as the sum of $\phi$ values for each subgroup, making it a sum of decreasing functions.





\paragraph{Targets Achievement} Finally, we demonstrate that the function $\phi$ reaches its lower bound \emph{iff} all agents have reached their goals.

It can be seen that, if all agents reach their targets, $\phi$ stops changing after the next time step, reaching its lower bound.

Conversely, assume $\phi$ is at its minimum, but at least one agent in a connected subgroup hasn't reached its target. If target assignments within this subgroup remain unchanged during lines~\ref{line:incons_begin}-\ref{line:incons_end} of Alg.~\ref{alg:tp_update}, $\phi_3$ stays constant while $\phi_1 + \phi_2$ decreases, leading to a contradiction.

If the target assignment changes, new goal information reduces $\phi_3$, decreasing $\phi$, which contradicts the assumption that $\phi$ has stopped changing.

\end{proof}

Note that while the proof idea is inspired by the completeness proof of the TSWAP algorithm, our scenario is significantly more complex due to its decentralized nature. To the best of our knowledge, this is the first time completeness has been proven in such a context.

\begin{figure}[!h]
    \centerline{\includegraphics[width=0.9\columnwidth]{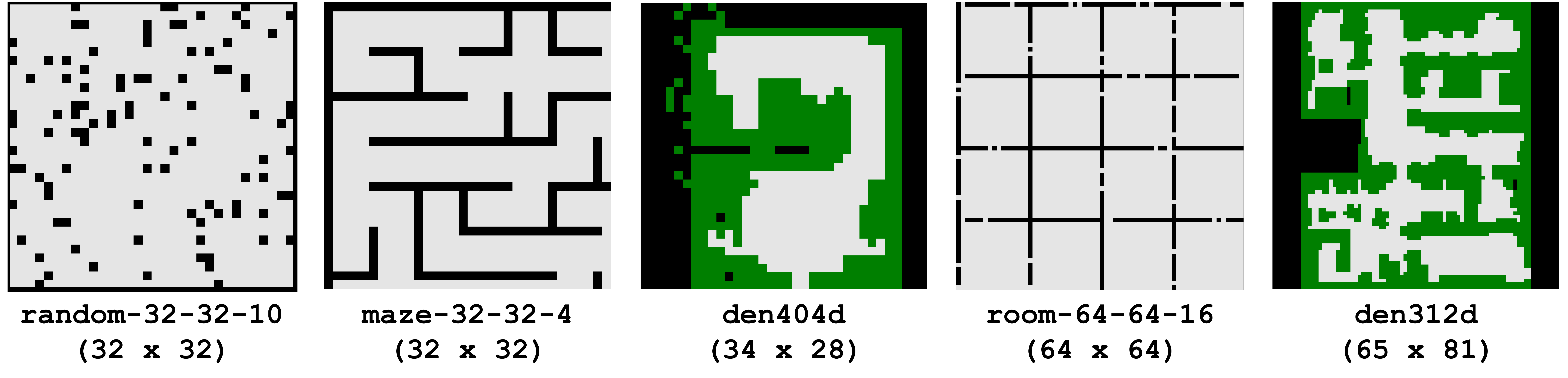}}
    \caption{Maps that are used in the experiments.} 
    \label{fig:maps}
\end{figure}
\vspace*{10pt}
\begin{figure}[!h]
    \centerline{\includegraphics[width=0.98\columnwidth]{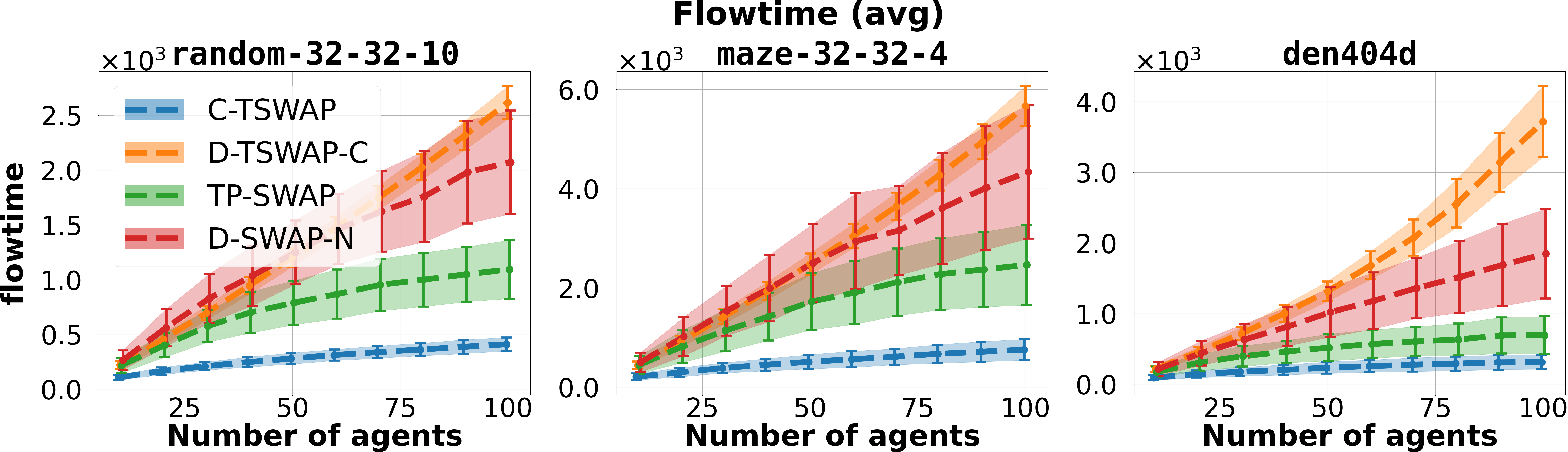}}
    \vspace{5pt}
    \centerline{\includegraphics[width=0.98\columnwidth]{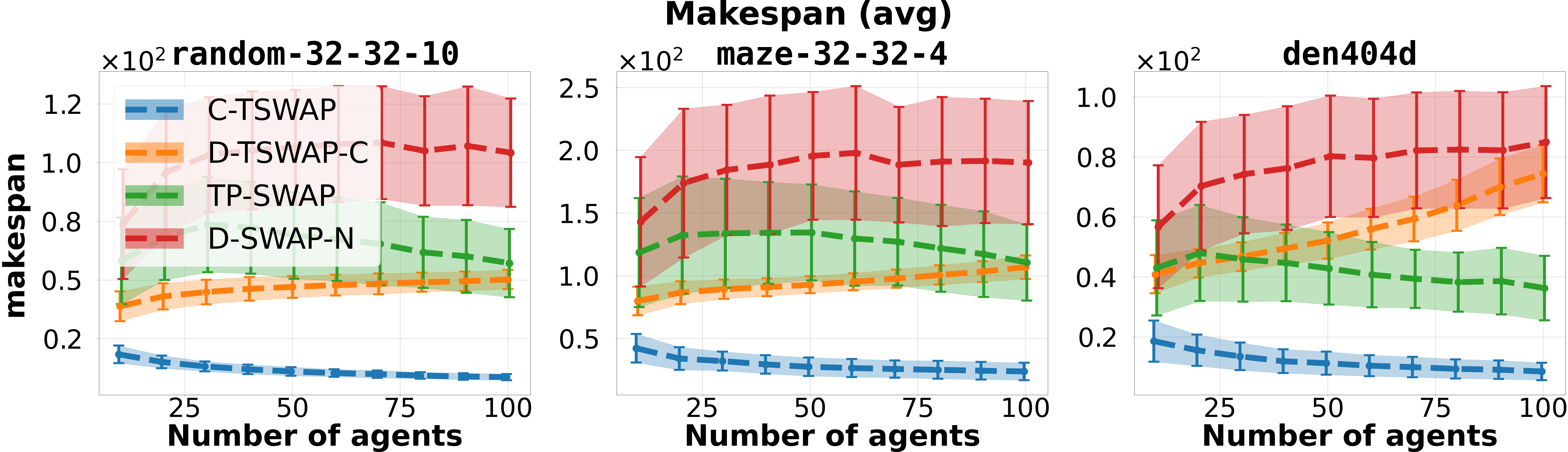}}
    \caption{Average \textit{\textbf{flowtime}}, \textit{\textbf{makespan}} and standard deviations (shaded ahead) of the evaluated AMAPF solvers. \vspace*{20 pt}} 
    \label{fig:flowtime_makespan}
\end{figure}

\section{Empirical Evaluation}
\label{sec:experiments}

\paragraph{Algorithms} Our evaluation considers the proposed decentralized AMAPF solvers and a centralized baseline, TSWAP. The latter is denoted as C-TSWAP (where ``C'' means centralized). In particular the following decentralized methods are evaluated: decentralized adaptation of TSWAP that relies on the initial random but consistent target assignment -- D-TSWAP-C; naive fully decentralized AMAPF solver that does not rely on consistent initial target assignment -- D-SWAP-N; its advanced variant that utilizes the suggested target and priority swapping procedure -- TP-SWAP. All decentralized methods are implemented by us\footnote{Source code: \url{https://github.com/PathPlanning/TP-SWAP}}. For C-TSWAP we used the original authors' implementation that utilize the bottleneck assignment (as this way of assigning targets was shown to perform better on average in the original paper).

\paragraph{Problem Instances} We utilized three grid maps from the MovingAI benchmark, which is well-known within the MAPF community~\citep{stern2019multi, sturtevant2012benchmarks}. The selected maps include \texttt{random-32-32-10}, \texttt{maze-32-32-4} and \texttt{den404d} (see Fig.~\ref{fig:maps}). They all have roughly the same size ($32 \times 32$) but differ significantly in topology.

For each map, we generate 250 different scenarios. Every scenario is a list of 100 start/target locations. To create an instance of $n$ agents from a scenario, we take first $n$ start-goal pairs from the list. In our experiments, we varied the number of agents from $10$ to $100$ with an increment of $10$. In total, for each map and each number of agents, we have 250 different problem instances.

The communication range for the decentralized algorithms was set to an area $5 \times 5$ cells with an agent in the center. The primary performance indicator we are interested in is the solution quality, measured as makespan and flowtime.

\paragraph{Makespan and Flowtime Metrics} Top row of Fig.~\ref{fig:flowtime_makespan} shows the average flowtime. The first important observation is that TP-SWAP, indeed, notably outperforms its naive decentralized counterpart, D-SWAP-N. The difference in their performance is getting more pronounced when the number of agents increases. On average, TP-SWAP is $2.3$ times better than D-SWAP-N across all the maps and numbers of agents. The standard deviation of flowtime values is also consistently smaller for TP-SWAP.

Interestingly, TP-SWAP also surpasses a semi-decentralized TSWAP variant with consistent initial target assignment, suggesting that TP-SWAP's initial assignment (where each agent picks the nearest target) is more effective, even if agents have to restore the assignment consistency.


The importance of the initial assignment is also exemplified by the performance of C-TSWAP (which is much better compared to the decentralized solvers). It confirms, that in case of smart centralized initial target assignment, one can achieve much better flowtime.

Regarding makespan (the bottom row of Fig.~\ref{fig:flowtime_makespan}), similar trends emerge, though TP-SWAP does not outperform D-TSWAP-C. This indicates that consistent target assignment has a stronger impact on makespan than flowtime. Notably, TP-SWAP's makespan nearly converges with D-TSWAP-C as the number of agents increases. Moreover, other algorithms (except D-TSWAP-C) also show decreased makespan with more agents.


This decrease in makespan may result from higher agent density, which aids in quicker recovery of a consistent assignment. C-TSWAP also benefits, as more agents allow finding closer targets initially. In contrast, D-TSWAP-C’s consistent assignment negates this effect, because its assignment is initially consistent, eliminating the need to restore consistency. Since the assignment is random, there is no advantage from goal proximity.


\paragraph{Impact of the Communication Range}

To investigate the impact of communication range on the performance of the proposed algorithm, we conducted a series of experiments on \texttt{maze-32-32-4} map using three different communication range sizes: $5 \times 5$, $11 \times 11$, and $21 \times 21$ cells.

\begin{table}[t]
    \centering
\begin{tabular}{c|ccc}
\toprule
{$n$} & $5 \times 5$ & $11 \times 11$ & $21 \times 21$ \\
\midrule
20  & 819  &508  & 414   \\
40  & 1423 &755  & 727   \\
60  & 1906 &994  & 989   \\
80  & 2279 &1297 & 1284  \\
100 & 2464 &1598 & 1560  \\
\bottomrule
\end{tabular}
    \vspace{10pt}
    \caption{Average \textit{\textbf{flowtime}} for TP-SWAP with varying communication ranges on the \texttt{maze-32-32-4} map.\\[5 pt]}
    \label{tab:ft_cr}
\end{table}

Table~\ref{tab:ft_cr} displays the flowtime values for varying communication ranges across different numbers of agents. The results demonstrate a significant performance boost when the communication radius increases from $5 \times 5$ to $11 \times 11$. However, this improvement diminishes when the radius is further extended to $21 \times 21$. These findings suggest that expanding the communication range enhances problem-solving efficiency, but only up to a certain point. Beyond this point, the overlap in communication ranges likely causes most agents to form a single, large connected group, meaning that further increases in the communication range do not substantially improve agent connectedness.

\begin{table}[!th]
    \centering
\begin{tabular}{c|cc|cc}
\toprule
{\textbf{Step}} &\multicolumn{2}{c|}{\texttt{den312d}} &\multicolumn{2}{c}{\texttt{room-64-64-16}} \\
{\textbf{limit}} &\textbf{TP-SWAP} &\textbf{D-TSWAP-N} &\textbf{TP-SWAP} &\textbf{D-TSWAP-N} \\

\midrule
600   &\textbf{100 \%}&\textbf{100 \%}&\textbf{100 \%}&\textbf{100 \%} \\
500   &\textbf{100 \%}&96 \% &\textbf{100 \%}&92 \% \\
400   &\textbf{94 \%}&67 \% &\textbf{92 \%}&62 \% \\
300   &\textbf{60 \%}&26 \% &\textbf{52 \%}&21 \% \\
200   &\textbf{11 \%}&4 \%&\textbf{12 \% }&2 \% \\
\bottomrule
\end{tabular}
    \\[10 pt]
    \caption{The \textit{\textbf{success rates}} of the fully decentralized AMAPF solvers under different timestep limits.\\[5 pt]}
    \label{tab:sr_dec}
\end{table}

\paragraph{Additional Comparison of the Fully Decentralized Solvers}

To get a more nuanced picture of how the performance of TP-SWAP differs from that of the basic fully-decentralized solver, D-TSWAP-N, we run additional experiments on two extra maps of slightly bigger size: \texttt{room-64-64-16} and \texttt{den312d} (see Fig.~\ref{fig:maps}). For each map, we generate 250 different instances involving 100 agents. Moreover, we introduce a timestep limit $T_{max}$ -- if the agents do not reach all the goals before the timestep $T_{max}$ we count this run as \emph{failure} (\emph{success} otherwise). We vary $T_{max}$ from 200 to 600 with an increment of 100. The results (success rate) are shown in Table~\ref{tab:sr_dec}.

As one can see, TP-SWAP solves a larger number of instances under any limit except 600 (when both methods solve all the tasks). Coupled with the results presented in the previous section, this confirms that the proposed target-priority swapping procedure is of utmost importance to the performance of the decentralized solvers.

\paragraph{Extended Results}

In the Supplementary material, extended empirical results are provided, including the ones on maps of different sizes and a more detailed analysis for various communication ranges, including statistics on the average number and size of connected subgroups.

\section{Discussion and Limitations}

This paper primarily addresses the theoretical aspects of decentralized multi-agent navigation, focusing on target selection and action choice, but several practical issues remain unaddressed.

Firstly, we assume that the agents possess all the necessary information about the others within their subgroup when choosing actions, which would require a specialized information-sharing mechanism in practice.

Secondly, we assume synchronized movements of agents. In the real world, e.g. when implementing our algorithm on robots, a decentralized motion synchronization might be needed, which, we believe, could be achieved via communication.

\section{Conclusion}


In this work, we have proposed a novel method to address the problem where a set of agents needs to reach a set of targets, and it does not matter which agent reaches a particular target. We focused on a particularly challenging and previously unsolved scenario in which the system is decentralized, allowing only local communication between the agents, and the initial goal assignment is inconsistent. To tackle this, we introduced an algorithm, TP-SWAP, specifically designed to solve this problem, and studied it both theoretically and empirically.

The experimental results demonstrated that TP-SWAP outperforms fully decentralized competitors in various scenarios and can achieve parity with, or even surpass, a semi-centralized solver that has access to consistent goal assignments, particularly in terms of flowtime. Future research directions include exploring more general AMAPF problem settings (e.g., colored MAPF), addressing communication issues, and implementing and evaluating our method on real robots.

\begin{ack}
The reported study was supported by the Ministry of Science and Higher Education of the Russian Federation under Project 075-15-2024-544.
\end{ack}



\bibliography{mybibfile}

\begin{thebibliography}{24}
\providecommand{\natexlab}[1]{#1}
\providecommand{\url}[1]{\texttt{#1}}
\expandafter\ifx\csname urlstyle\endcsname\relax
  \providecommand{\doi}[1]{doi: #1}\else
  \providecommand{\doi}{doi: \begingroup \urlstyle{rm}\Url}\fi

\bibitem[Ali and Yakovlev(2024)]{ali2024improved}
Z.~A. Ali and K.~Yakovlev.
\newblock Improved {{Anonymous Multi-Agent Path Finding Algorithm}}.
\newblock In \emph{Proceedings of the {{AAAI Conference}} on {{Artificial
  Intelligence}}}, volume~38, pages 17291--17298, 2024.

\bibitem[Barer et~al.(2014)Barer, Sharon, Stern, and
  Felner]{barer2014suboptimal}
M.~Barer, G.~Sharon, R.~Stern, and A.~Felner.
\newblock Suboptimal variants of the conflict-based search algorithm for the
  multi-agent pathfinding problem.
\newblock In \emph{Proceedings of The 7th Annual Symposium on Combinatorial
  Search ({SoCS} 2014)}, pages 19--27, 2014.

\bibitem[{\v{C}}{\'a}p et~al.(2015){\v{C}}{\'a}p, Nov{\'a}k, Kleiner, and
  Seleck{\`y}]{vcap2015prioritized}
M.~{\v{C}}{\'a}p, P.~Nov{\'a}k, A.~Kleiner, and M.~Seleck{\`y}.
\newblock Prioritized planning algorithms for trajectory coordination of
  multiple mobile robots.
\newblock \emph{IEEE Transactions on Automation Science and Engineering},
  12\penalty0 (3):\penalty0 835--849, 2015.

\bibitem[de~Wilde et~al.(2013)de~Wilde, ter Mors, and Witteveen]{de2013push}
B.~de~Wilde, A.~W. ter Mors, and C.~Witteveen.
\newblock Push and rotate: cooperative multi-agent path planning.
\newblock In \emph{Proceedings of the 12th International Conference on
  Autonomous Agents and Multiagent Systems ({AAMAS} 2013)}, pages 87--94, 2013.

\bibitem[Ferner et~al.(2013)Ferner, Wagner, and Choset]{ferner2013odrm}
C.~Ferner, G.~Wagner, and H.~Choset.
\newblock Odrm* optimal multirobot path planning in low dimensional search
  spaces.
\newblock In \emph{Proceedings of The 2013 IEEE International Conference on
  Robotics and Automation ({ICRA} 2013)}, pages 3854--3859. IEEE, 2013.

\bibitem[Ji et~al.(2021)Ji, Li, Pan, Gao, and Tu]{ji2021decentralized}
X.~Ji, H.~Li, Z.~Pan, X.~Gao, and C.~Tu.
\newblock Decentralized, unlabeled multi-agent navigation in obstacle-rich
  environments using graph neural networks.
\newblock In \emph{Proceedings of IEEE/RSJ International Conference on
  Intelligent Robots and Systems (IROS 2021)}, pages 8936--8943, 2021.

\bibitem[Khan et~al.(2019)Khan, Zhang, Li, Wu, Schlotfeldt, Tang, Ribeiro,
  Bastani, and Kumar]{khan2019learning}
A.~Khan, C.~Zhang, S.~Li, J.~Wu, B.~Schlotfeldt, S.~Y. Tang, A.~Ribeiro,
  O.~Bastani, and V.~Kumar.
\newblock Learning safe unlabeled multi-robot planning with motion constraints.
\newblock In \emph{Proceedings of IEEE/RSJ International Conference on
  Intelligent Robots and Systems (IROS 2019)}, pages 7558--7565, 2019.

\bibitem[Khan et~al.(2021)Khan, Kumar, and Ribeiro]{khan2021large}
A.~Khan, V.~Kumar, and A.~Ribeiro.
\newblock Large scale distributed collaborative unlabeled motion planning with
  graph policy gradients.
\newblock \emph{IEEE Robotics and Automation Letters}, 6\penalty0 (3):\penalty0
  5340--5347, 2021.

\bibitem[Li et~al.(2021)Li, Ruml, and Koenig]{li2021eecbs}
J.~Li, W.~Ruml, and S.~Koenig.
\newblock Eecbs: A bounded-suboptimal search for multi-agent path finding.
\newblock In \emph{Proceedings of the 35th {AAAI} Conference on Artificial
  Intelligence ({AAAI} 2021)}, pages 12353--12362, 2021.

\bibitem[Lowe et~al.(2017)Lowe, Wu, Tamar, Harb, Pieter~Abbeel, and
  Mordatch]{lowe2017multi}
R.~Lowe, Y.~I. Wu, A.~Tamar, J.~Harb, O.~Pieter~Abbeel, and I.~Mordatch.
\newblock Multi-agent actor-critic for mixed cooperative-competitive
  environments.
\newblock \emph{Proceedings of the Advances in neural information processing
  systems ({NIPS} 2017)}, 30, 2017.

\bibitem[Okumura and D{\'e}fago(2023)]{okumura2023solving}
K.~Okumura and X.~D{\'e}fago.
\newblock Solving simultaneous target assignment and path planning efficiently
  with time-independent execution.
\newblock \emph{Artificial Intelligence}, 321:\penalty0 103946, 2023.

\bibitem[Okumura et~al.(2022)Okumura, Machida, D{\'e}fago, and
  Tamura]{okumura2022priority}
K.~Okumura, M.~Machida, X.~D{\'e}fago, and Y.~Tamura.
\newblock Priority inheritance with backtracking for iterative multi-agent path
  finding.
\newblock \emph{Artificial Intelligence}, 310:\penalty0 103752, 2022.

\bibitem[Panagou et~al.(2019)Panagou, Turpin, and
  Kumar]{panagou2019decentralized}
D.~Panagou, M.~Turpin, and V.~Kumar.
\newblock Decentralized goal assignment and safe trajectory generation in
  multirobot networks via multiple lyapunov functions.
\newblock \emph{IEEE Transactions on Automatic Control}, 65\penalty0
  (8):\penalty0 3365--3380, 2019.

\bibitem[Sartoretti et~al.(2019)Sartoretti, Kerr, Shi, Wagner, Kumar, Koenig,
  and Choset]{sartoretti2019primal}
G.~Sartoretti, J.~Kerr, Y.~Shi, G.~Wagner, T.~S. Kumar, S.~Koenig, and
  H.~Choset.
\newblock Primal: Pathfinding via reinforcement and imitation multi-agent
  learning.
\newblock \emph{IEEE Robotics and Automation Letters}, 4\penalty0 (3):\penalty0
  2378--2385, 2019.

\bibitem[Sharon et~al.(2013)Sharon, Stern, Goldenberg, and
  Felner]{sharon2013increasing}
G.~Sharon, R.~Stern, M.~Goldenberg, and A.~Felner.
\newblock The increasing cost tree search for optimal multi-agent pathfinding.
\newblock \emph{Artificial intelligence}, 195:\penalty0 470--495, 2013.

\bibitem[Sharon et~al.(2015)Sharon, Stern, Felner, and
  Sturtevant.]{sharon2015conflict}
G.~Sharon, R.~Stern, A.~Felner, and N.~R. Sturtevant.
\newblock Conflict-based search for optimal multiagent path finding.
\newblock \emph{Artificial Intelligence Journal}, 218:\penalty0 40--66, 2015.

\bibitem[Skrynnik et~al.(2024)Skrynnik, Andreychuk, Nesterova, Yakovlev, and
  Panov]{skrynnik2024learn}
A.~Skrynnik, A.~Andreychuk, M.~Nesterova, K.~Yakovlev, and A.~Panov.
\newblock Learn to {{Follow}}: {{Decentralized Lifelong Multi-Agent
  Pathfinding}} via {{Planning}} and {{Learning}}.
\newblock In \emph{Proceedings of the {{AAAI Conference}} on {{Artificial
  Intelligence}}}, volume~38, pages 17541--17549, 2024.

\bibitem[Stern et~al.(2019)Stern, Sturtevant, Felner, Koenig, Ma, Walker, Li,
  Atzmon, Cohen, Kumar, Boyarski, and Bartak]{stern2019multi}
R.~Stern, N.~R. Sturtevant, A.~Felner, S.~Koenig, H.~Ma, T.~T. Walker, J.~Li,
  D.~Atzmon, L.~Cohen, T.~K.~S. Kumar, E.~Boyarski, and R.~Bartak.
\newblock Multi-agent pathfinding: Definitions, variants, and benchmarks.
\newblock \emph{Symposium on Combinatorial Search (SoCS)}, pages 151--158,
  2019.

\bibitem[Sturtevant(2012)]{sturtevant2012benchmarks}
N.~Sturtevant.
\newblock Benchmarks for grid-based pathfinding.
\newblock \emph{Transactions on Computational Intelligence and AI in Games},
  4\penalty0 (2):\penalty0 144 -- 148, 2012.
\newblock URL \url{http://web.cs.du.edu/~sturtevant/papers/benchmarks.pdf}.

\bibitem[Turpin et~al.(2014)Turpin, Michael, and Kumar]{turpin2014capt}
M.~Turpin, N.~Michael, and V.~Kumar.
\newblock Capt: Concurrent assignment and planning of trajectories for multiple
  robots.
\newblock \emph{The International Journal of Robotics Research}, 33\penalty0
  (1):\penalty0 98--112, 2014.

\bibitem[Wagner and Choset(2011)]{wagner2011m}
G.~Wagner and H.~Choset.
\newblock M*: {A} complete multirobot path planning algorithm with performance
  bounds.
\newblock In \emph{Proceedings of The 2011 {IEEE/RSJ} International Conference
  on Intelligent Robots and Systems ({IROS} 2011)}, pages 3260--3267, 2011.

\bibitem[Wang et~al.(2020)Wang, Liu, Li, and Prorok]{wang2020mobile}
B.~Wang, Z.~Liu, Q.~Li, and A.~Prorok.
\newblock Mobile robot path planning in dynamic environments through globally
  guided reinforcement learning.
\newblock \emph{IEEE Robotics and Automation Letters}, 5\penalty0 (4):\penalty0
  6932--6939, 2020.

\bibitem[Wang et~al.(2023)Wang, Xiang, Huang, and Sartoretti]{wang2023scrimp}
Y.~Wang, B.~Xiang, S.~Huang, and G.~Sartoretti.
\newblock Scrimp: Scalable communication for reinforcement-and
  imitation-learning-based multi-agent pathfinding.
\newblock In \emph{Proceedings of The 2023 {IEEE/RSJ} International Conference
  on Intelligent Robots and Systems ({IROS} 2023)}, pages 9301--9308, 2023.

\bibitem[Yu and LaValle(2013)]{yu2013multi}
J.~Yu and S.~M. LaValle.
\newblock Multi-agent path planning and network flow.
\newblock In \emph{Algorithmic Foundations of Robotics X: Proceedings of the
  Tenth Workshop on the Algorithmic Foundations of Robotics}, pages 157--173.
  Springer, 2013.

\end{thebibliography}


\appendix

\newpage

\begin{center}
    \Large{\textbf{Supplementary Material}}
\end{center}
\section{Extended Experiments}

\begin{figure}[!h]
    \centerline{\includegraphics[width=0.9\columnwidth]{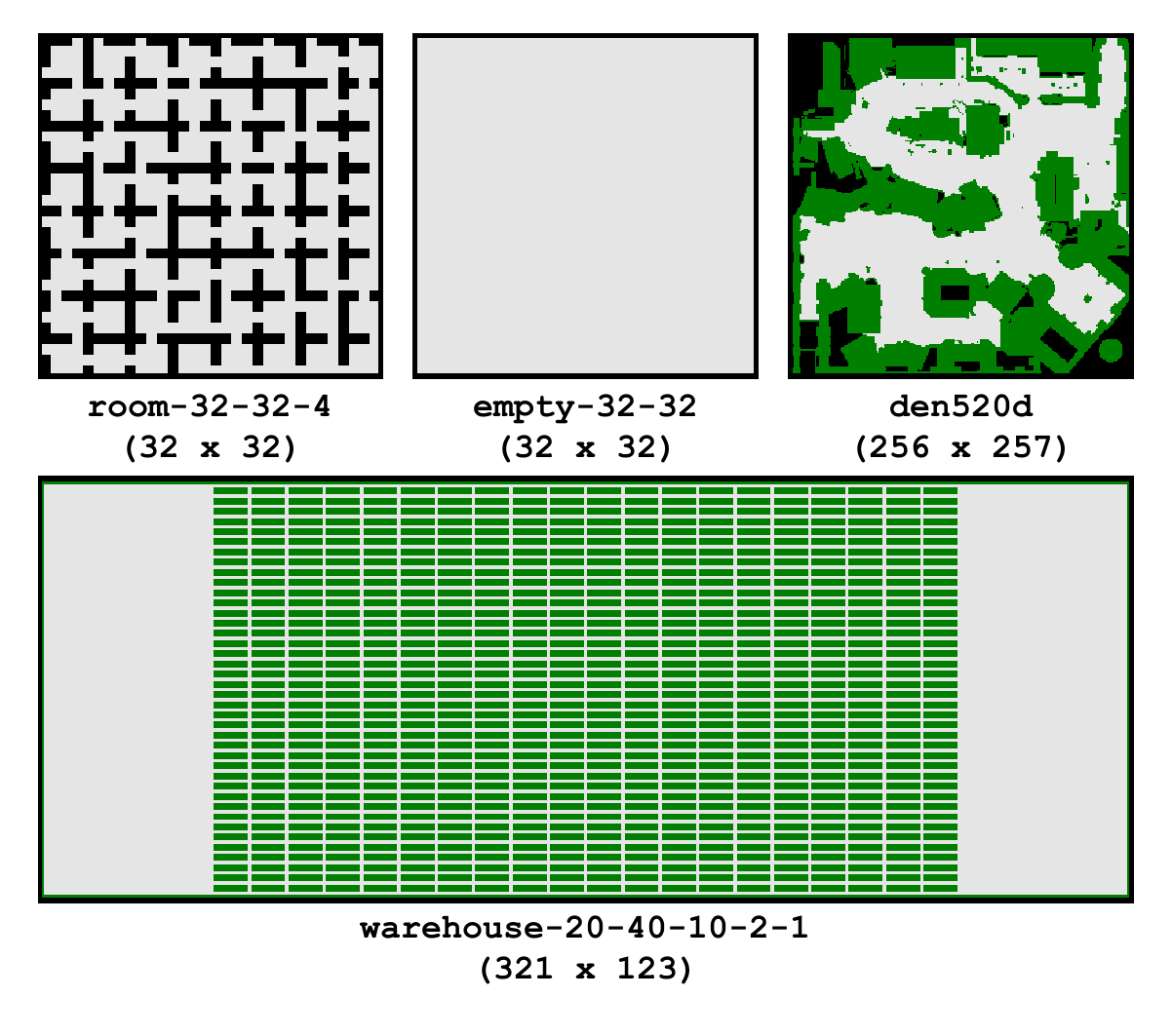}}
    \caption{Maps that are used in the extended experiments.} 
    \label{fig:maps_ext}
\end{figure}

\begin{figure*}[!t]
    \centerline{\includegraphics[width=0.9\textwidth]{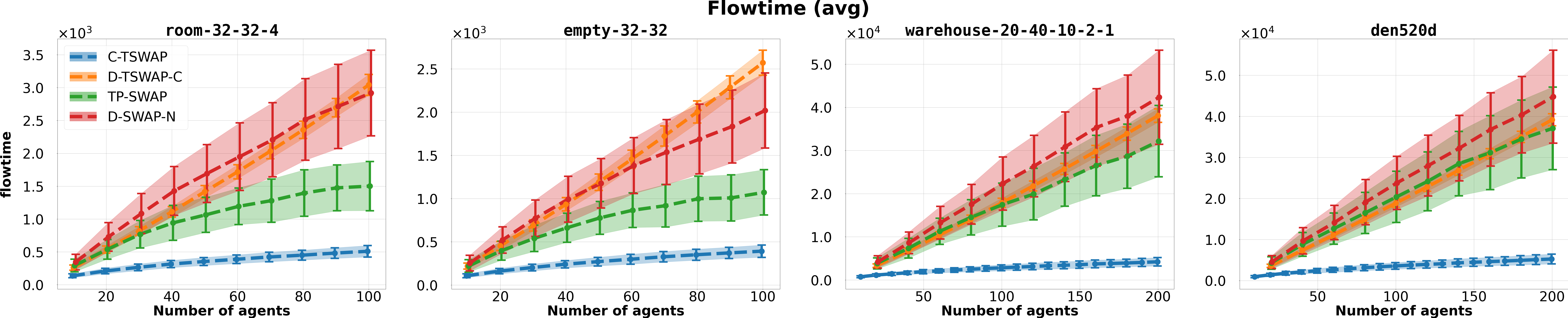}}
    \caption{Average \textit{\textbf{flowtime}} and its standard deviation (shaded ahead) of the evaluated AMAPF solvers on an extended set of maps} 
    \label{fig:flowtime_ext}
\end{figure*}

\begin{figure*}[!t]
    \centerline{\includegraphics[width=0.9\textwidth]{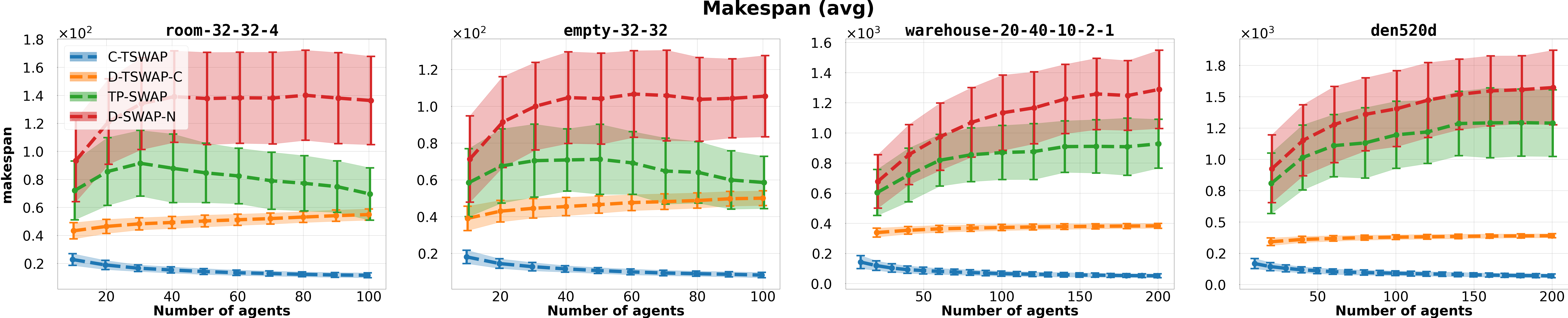}}
    \caption{Average \textit{\textbf{makespan}} and its standard deviation (shaded ahead) of the evaluated AMAPF solvers on an extended set of maps} 
    \label{fig:makespan_ext}
\end{figure*}

\subsection{Further Evaluation On Additional Maps}

To further validate the performance and scalability of the proposed decentralized algorithm, we conducted an additional series of experiments on two small (\texttt{random-32-32-10}, \texttt{empty-32-32}) and two large (\texttt{warehouse-20-40-10-2-1}, \texttt{den520d}) maps from the MovingAI benchmark~\citep{stern2019multi}.

Similarly to the experiments reported in the main body, we generated 250 different scenarios for each map, each scenario containing 100 start/target pairs on small maps and 200 start/target pairs on the large ones. The number of agents varied from 10 to 100 (20 to 200), with increments of 10 (20). The communication range for the decentralized algorithms was set to 5x5 cells with an agent in the center (as before).

The results are presented in Fig.~\ref{fig:flowtime_ext}-\ref{fig:makespan_ext}. They generally align with those described in Section~\ref{sec:experiments}. However, on the large maps, the difference between decentralized and centralized algorithms becomes more pronounced as the map size increases. Notably, the algorithm with the consistent initial goal assignment (D-TSWAP-C) shows superiority in makespan compared to the fully decentralized one (TP-SWAP). Despite this, the total solution duration (flowtime) remains similar across the solvers. Furthermore, as the number of agents increases, TP-SWAP begins to outperform D-TSWAP-C in terms of flowtime.

These effects can be attributed to the fact that the larger map size significantly complicates the coordination among decentralized agents. The inability to quickly recover a consistent assignment leads to the situations where some agents must visit multiple occupied targets before eventually finding an available one. This inefficiency contributes to widening of the performance gap between decentralized and centralized approaches on the larger maps.

Overall, the conducted additional experiments confirm that the proposed method is robust and capable of functioning effectively across various environments, including larger-scale settings. However, the size of the environment and the density of agents can notably impact the quality of its solutions, particularly when compared to the centralized algorithm.

\begin{table}[t]
    \centering
\begin{tabular}{c|ccc|ccc}
\toprule
{} & \multicolumn{3}{c|}{\textit{\textbf{Number of groups}}} & \multicolumn{3}{c}{\textit{\textbf{Groups' sizes}}} \\
{$n$} & $5 \times 5$ & $11 \times 11$ & $21 \times 21$ & $5 \times 5$ & $11 \times 11$ & $21 \times 21$\\
\midrule
10     &            9 &            6 &             2 &                1 &            2 &             6 \\
20     &           16 &            7 &             1 &                1 &            3 &            18 \\
30     &           21 &            5 &             1 &                1 &            7 &            29 \\
40     &           24 &            3 &             1 &                2 &           15 &            40 \\
50     &           27 &            2 &             1 &                2 &           29 &            50 \\
60     &           28 &            2 &             1 &                2 &           45 &            60 \\
70     &           28 &            1 &             1 &                3 &           60 &            70 \\
80     &           28 &            1 &             1 &                3 &           74 &            80 \\
90     &           27 &            1 &             1 &                3 &           86 &            90 \\
100    &           25 &            1 &             1 &                4 &           97 &           100 \\
\bottomrule
\end{tabular}
    \\[10 pt]
    \caption{Average \textit{\textbf{number of subgroups}} and average \textit{\textbf{subgroup size}} during task execution by the TP-SWAP algorithm on the \texttt{maze-32-32-4} map.\\[5 pt]}
    \label{tab:mg_mgs_cr}
\end{table}

\begin{table*}[!t]
    \centering
\begin{tabular}{c|cccc|cccc}
\toprule
{} & \multicolumn{4}{c|}{\textbf{\textit{makespan}}} & \multicolumn{4}{c}{\textbf{\textit{flowtime}}} \\

{$n$} & $5 \times 5$ & $11 \times 11$ & $21 \times 21$ & D-TSWAP-C & $5 \times 5$ & $11 \times 11$ & $21 \times 21$ & D-TSWAP-C \\
\midrule
10     &          118 &           87 &            68 &        80 &          445 &          340 &           266 &       439 \\
20     &          132 &           83 &            68 &        87 &          819 &          508 &           414 &       920 \\
30     &          134 &           78 &            72 &        89 &         1139 &          642 &           586 &      1419 \\
40     &          134 &           73 &            72 &        91 &         1423 &          755 &           727 &      1933 \\
50     &          134 &           74 &            73 &        93 &         1725 &          889 &           864 &      2479 \\
60     &          130 &           71 &            71 &        95 &         1906 &          994 &           989 &      3045 \\
70     &          127 &           73 &            74 &        98 &         2118 &         1130 &          1143 &      3643 \\
80     &          122 &           74 &            73 &       101 &         2279 &         1297 &          1284 &      4273 \\
90     &          117 &           75 &            74 &       103 &         2372 &         1437 &          1424 &      4944 \\
100    &          110 &           76 &            75 &       107 &         2464 &         1598 &          1560 &      5665 \\
\bottomrule
\end{tabular}
    \\[10 pt]
    \caption{Average \textit{\textbf{makespan}} and \textit{\textbf{flowtime}} for TP-SWAP with varying communication range sizes on the \texttt{maze-32-32-4} map. For comparison, results from TSWAP with a random consistent initial assignment (D-TSWAP-C) are also presented.\\[5 pt]}
    \label{tab:ms_ft_cr}
\end{table*}

\subsection{Additional Evaluation Of The Impact Of Varying Communication Range}

Tables~\ref{tab:mg_mgs_cr} and \ref{tab:ms_ft_cr} presents additional details of the experiment involving variation of the communication range. Table~\ref{tab:ms_ft_cr} shows the \textit{\textbf{makespan}}, \textit{\textbf{flowtime}} and Table~\ref{tab:mg_mgs_cr} contains statistics related to the subgroups of agents for each communication range across different numbers of agents. Additionally, Table~\ref{tab:mg_mgs_cr} includes results for the TSWAP algorithm with a consistent random initial assignment, denoted as D-TPSWAP-C.

The results for the makespan and flowtime are consistent with those presented in Section~\ref{sec:experiments}. Notably, the proposed method surpasses the D-TPSWAP-C approach once a certain agent density threshold is reached. This advantage arises because TP-SWAP initially selects targets based on proximity, whereas D-TPSWAP-C assigns targets randomly, potentially leading to greater initial distances between agents and their targets. As agent density increases, decentralized agents in TP-SWAP can rapidly re-establish a consistent assignment, often reaching better targets and thereby outperforming D-TPSWAP-C.

Examining the statistics on the average number of subgroups and the average number of agents within these subgroups (Table~\ref{tab:mg_mgs_cr}), we observe that with a communication range of $5 \times 5$, even with 100 agents on the relatively small $32 \times 32$ map, the agents do not consolidate into a single large group. Instead, they form multiple smaller subgroups. When fewer agents are present on the map, they tend to operate largely independently, only occasionally exchanging information. Despite this limited communication, the algorithm effectively solves the problem, as demonstrated by the results. Remarkably, it competes well with the partially centralized D-TPSWAP-C method, even under these conditions.

On the other hand, increasing the communication radius facilitates full coordination among agents across the map, often leading to the formation of a single large connected group. However, it is important to note that even with a communication range of $11 \times 11$ and 100 agents on the map, the average number of agents in a group does not equal the total number of agents. This suggests that even in these scenarios, some agents occasionally operate independently and without constant communication with the rest.

\end{document}